\documentclass[letterpaper]{IEEEtran}
\usepackage{multicol}
\usepackage{lipsum}
\usepackage{array}
\usepackage[english]{babel}
\usepackage{amsmath,amssymb}
\usepackage{times}
\usepackage{relsize}
\usepackage{graphicx}
\usepackage{url}
\usepackage{booktabs}
\usepackage{multirow}
\usepackage{mparhack}
\usepackage{subfigure}
\usepackage{authblk}
\usepackage{amsthm}

\usepackage[noend]{algorithmic}
\usepackage[linesnumbered,ruled,vlined]{algorithm2e}

\usepackage{array}
\usepackage{cite}
\usepackage{eqparbox}
\usepackage{mdwmath}
\usepackage{balance}
\usepackage{epsfig}
\usepackage{xcolor}

\usepackage{mathtools}
\usepackage{bm}
\usepackage{amsfonts}
\usepackage[all]{xy}
\usepackage{etoolbox}
\usepackage{graphicx}
\DeclareMathAlphabet{\mathantt}{OT1}{antt}{li}{it}
\DeclareMathAlphabet{\mathpzc}{OT1}{pzc}{m}{it}

\usepackage{accents}

\newtheorem{theorem}{Theorem}

\newtheorem{lemma}[theorem]{Lemma}

\DeclareFontFamily{OT1}{pzc}{}
\DeclareFontShape{OT1}{pzc}{m}{it}%
  {<-> s * [1.1] pzcmi7t}{}
\DeclareMathAlphabet{\mathpzc}{OT1}{pzc}%
                     {m}{it}

\DeclareMathOperator{\argmin}{\arg\min}

\makeatletter
\patchcmd{\@maketitle}
  {\addvspace{0.8\baselineskip}\egroup}
  {\addvspace{-1.45\baselineskip}\egroup}
  {}
  {}
\makeatother

\title{On Optimal Proactive and Retention-Aware Caching with User Mobility}
\author{Ghafour Ahani~and~\and Di Yuan}
\affil{Department of Information Technology \\ Uppsala University, Sweden\\
Emails:\{ghafour.ahani, di.yuan\}@it.uu.se
}
\begin{document}
\pagenumbering{gobble}
\maketitle
\begin{abstract}
Caching popular contents at edge devices is an effective solution to alleviate the burden of the backhaul networks.  Earlier investigations commonly neglected the storage cost in caching. More recently, retention-aware caching, where both the downloading cost and storage cost are accounted for, is attracting attention. Motivated by this, we address proactive and retention-aware caching problem with the presence of user mobility, optimizing the sum of the two types of costs. More precisely, a cost-optimal caching problem for vehicle-to-vehicle networks is formulated with joint consideration of the impact of the number of vehicles, cache size, storage cost, and content request probability. This is a combinatorial optimization problem. However, we derive a stream of analytical results and they together lead to an algorithm that guarantees global optimum with polynomial-time complexity. Numerical results show significant improvements in comparison to popular caching and random caching.

\end{abstract}

\IEEEpeerreviewmaketitle
\vskip -15pt
\section{Introduction}
The explosive mobile data traffic growth is putting a heavy burden on backhaul links, causing delays in downloading contents. However, a large portion of the mobile traffic is due to duplicate downloads of a few popular contents. Caching technology has been considered as an effective solution to reduce the burden of the backhaul, by storing the contents on edge devices~\cite{XWang2014Cache}. This enables the mobile users to download their requested contents from the nearby devices instead of downloading the contents from the core network. In modeling such scenarios, most of research efforts focused on downloading cost.  However, storing a content may be subject to a cost as well. A storage cost may be due to flash rental cost incurred by cloud service providers or flash damage caused by writing a content to the memory device~\cite{Shukla2017TMC}. In both cases, the storage cost typically depends on the time duration of storage, hereafter referred to as the retention time. Intuitively, with longer retention time, the requested contents can be obtained with higher probability from the cache, thus avoiding the cost of downloading from the network. But longer retention time results in a higher storage cost. Therefore, \emph{what to cache} and \emph{for how long} are both key aspects in optimal caching. Few works in studying optimal caching have considered the impact of storage cost. The works such as~\cite{KPoularakis2016, TDeng2017Cost,taodeng2018} considered only the downloading cost. The study in \cite{KPoularakis2016} proposed an approximation algorithm with performance guarantee for multicast-aware proactive caching. The authors in \cite{TDeng2017Cost} considered cost-optimal caching with user mobility. They presented an extension in \cite{taodeng2018} by providing a linear lower bound of the objective function. In these works, storage cost was neglected.
The study in~\cite{Abedini2014Content} chose to represent storage cost using a random variable. Later, the work in \cite{Schroeder2016Flash} suggested that the storage cost can be better modeled by an increasing linear/convex function. Another limitation of \cite{Abedini2014Content} is that the retention time is fixed. Later, this assumption was relaxed in \cite{Shukla2017Proactive} and the retention time was treated as an optimization variable in a time-slotted system. We remark that in \cite{Shukla2017Proactive}, a user is associated with only one cache. A generalization of a multiple-path routing model with retention-aware caching was studied in \cite{Shukla2017Hold}.

For mobility scenarios, contents are often cached at mobile devices. They can exchange the requested contents when they move into the communication range of each other. Making the best of mobility information between mobile users can significantly improve the caching efficiency~\cite{TDeng2017Cost,RWang2016Mobility,RWang2016}.  However, considering both downloading and storage costs, with presence of user mobility, calls for further research. To this end, our main contributions are as follows.

\begin{itemize}
\item
We formulate a Proactive Retention-Aware Caching Optimization (PRACO) problem with user mobility in a time-slotted system, taking into account both the downloading and storage costs.
\item This problem is a combinational optimization in its nature. However, we provide mathematical analysis in order to facilitate the computation of global optimum time-efficiently, namely,
\begin{itemize}
\item
we first prove that for any content, the optimal caching decisions over the time slots can be derived, given the initial number of mobile users caching the content;
\item
the above analysis is then embedded into a dynamic programming algorithm and we prove it is both globally optimal and of polynomial-time complexity.
\end{itemize}
\item
The numerical results show significant improvements in comparison to two conventional algorithms, namely, popular caching and random caching.
\end{itemize}


\section{System Model and Problem Formulation}
\label{System_model}
\subsection{System Model}

We consider a vehicular network scenario which consists of a content server having all the contents, a number of vehicles, and road side units (RSUs) providing signal coverage for the vehicles. Denote by $\mathcal{R}$ the set of vehicles that are interested in requesting contents, referred to as \emph{requesters}, whose index set is represented by $\mathcal{R}=\{1,2,\dots,R\}$. Denote by $\mathcal{H}$ the set of vehicles that we call \emph{helpers}. Each helper is equipped with a cache of size $s$, that can supply the requesters with contents from the cache, and therefore to mitigate backhaul congestion. We consider a library of $C$ contents, whose index set is $\mathcal{C}=\{1,2,\dots,C \}$. The sizes of all the contents are the same and are assumed to be one. In addition, each content is either fully stored or not stored at all at a helper. Figure \ref{fig:systemmodel} shows the system scenario.

\begin{figure}[ht!]
\centering
\includegraphics[scale=0.4]{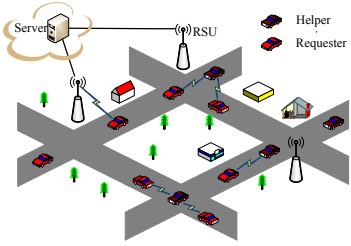}
\vspace{-3mm}
\begin{center}
  \caption{System scenario.}
  \label{fig:systemmodel}
  \vspace{-3mm}
\end{center}
\end{figure}
\vspace{-5mm}

The event that vehicles move into the transmission range of each other is called a \emph{contact}, during which communication between them can occur.  We consider that the contact between any two vehicles follows a Poisson distribution. Poisson distribution can characterize the mobility pattern of vehicles as the tail behavior of the inter-contact time distribution can be characterized as an exponential distribution by analyzing the real-world vehicle mobility traces, see~\cite{Zhu2010Recognizing}.  Here, we assume a homogeneous contact rate for all the vehicles, denoted by $\lambda$. This assumption is in fact common~\cite{Spyropoulos2008Efficient,Zhang2007Performance}.
As a consequence, it is not necessary to explicitly consider the content cached by each helper, as there is no difference between the helpers in the perspective of the requesters. Thus, the caching performance is fully determined by the number of helpers for each content. Moreover, it is obvious that a content is cached by no more than $H$ helpers. Therefore, in modeling cache capacity, it is sufficient to constrain that the total number of the cached contents of all the helpers does not exceed $\emph{S}=sH$.

We consider a time-slotted system where each time slot\footnote{Note
that the time duration of a slot is different from that in LTE. Here,
the order of magnitude in performance evaluations is hour.} is of duration $\delta$. The time period subject to optimization consists of $T$ time slots. In each slot, all the requesters are active and ask for some content. Also, no requester becomes helper in the next slots or vice versa. Each requester has its content request probabilities, of which the distribution is independent of time slot. The probability that content $c$ is requested by requester $r$ is denoted by $w_{rc}$, with $\sum_{c \in \mathcal{C}}w_{rc}=1$.
The contents, if cached by the helpers, are fetched at the beginning of the time period. When a requester asks for a content in a slot, the requester will first try to collect the content from the encountered helpers. If the requester fails to obtain the content at the end of this slot, it downloads the content from the server. In the latter case, a downloading cost is incurred.

\subsection{Cost Model}

Denote by $\bm{x}$ our caching decision which is a $C\times T$ matrix
for the $C$ contents and the $T$ slots.  The entry at location $(c,t)$, i.e., $x_{ct}$, denotes the number of helpers storing content $c$ in slot $t$, $x_{ct} \in \{0,1,\dots,H\}$.  The caching optimization process, applied at the beginning of the time period, will determine the number of helpers for each content as well as the
retention time. The latter is represented by the number of helpers
over the time slots, and this number either remains or decreases from
one time slot to the next.

Downloading a content from the server results in a
downloading cost. Also, caching a content in a helper has a
storage cost due to storage rental cost and flash memory damage. Same as~\cite{Shukla2017Proactive} and~\cite{Shukla2017Hold}, we neglect
the cost of the helpers to fill their caches at the beginning of the
time period. In addition, for the requesters, the downloading cost
from helpers is negligible in comparison to that from
the server.  Therefore, the total cost consists of the downloading
cost from the server for the requesters and the storage cost for the
helpers.

Denote by $f(t)$ the storage cost due to storing a content in a
helper's cache in slot $t$. A longer retention time needs a higher
threshold voltage, which results in a higher memory damage and
consequently gives a higher storage cost, for more detailed
discussions, see~\cite{Shukla2017TMC}. Motivated by this, we assume
that $f(t)$ is an increasing function.

When content $c$ is requested by requester $r$ in slot $t$, the
probability that the requester has to download the content from the
server is denoted by $p_{crt}$.  If $r$ does not meet any helper
having $c$ within the slot, the only way of obtaining $c$ is to
download from the server. As the contacts between the users follows a
Poisson distribution, $p_{crt}$ is given by:

\vspace{-3mm}
\[
p_{crt}=e^{-x_{ct}\lambda \delta}
\]
\vspace{-5mm}

Thus, the total cost, denoted by $\text{Cost}(\bm{x})$, reads as:

\vspace{-5mm}
\begin{equation}
\begin{aligned}
\text{Cost}(\bm{x})=\underbrace{\sum_{r\in \mathcal{R}}\sum_{c \in \mathcal{C}} \sum_{t=1}^{T}w_{rc}p_{crt}}_{\text{downloading cost}}+\alpha\underbrace{ \sum_{c\in \mathcal{C}}\sum_{t=1}^{T}f(t)x_{ct}}_{\text{storage cost}},
\end{aligned}
\end{equation}
\vspace{-1mm}
where $\alpha$ is the weighting factor of the two cost types.

\subsection{Problem Formulation}

The proactive retention-aware caching optimization (PRACO) problem can be formulated as~\eqref{eq:formulation}.

\vspace{-2mm}
\begin{figure}[!h]
\vskip -10pt
\begin{subequations}
\begin{alignat}{2}
\quad &
\min\limits_{\bm{x}}\quad  \text{Cost}(\bm{x})
\label{F_e} \\
\text{s.t}. \quad
& \underset{c\in \mathcal{C}}\sum x_{ct} \le S,~\forall t\in \{1,2,\dots,T\} \label{totalcachesize}\\
& x_{ct}\in \{0,1,2,\dots,H\},c\in \mathcal{C},t\in \{1,2,\dots,T\}
\end{alignat}
\label{eq:formulation}
\vskip -20pt
\end{subequations}
\end{figure}

Constraints~(\ref{totalcachesize}) guarantees that the total number of
stored contents in each time slot does not exceed the total cache
capacity, i.e., $S$. The formulation does not explicitly require the
number of helpers of any content does not increase over time. This
aspect is analyzed later on in Section~\ref{sec:analysis}.

\section{Problem Analysis}
\label{sec:analysis}

We prove that for each content, the optimal number of helpers caching the content decreases over the time slots.  Next, we present an algorithm, which, with respect to the possible initial numbers of helpers of a content, computes the optimal number of helpers for this content over time. We then prove that the algorithm's optimality and its polynomial time complexity.

\begin{lemma}
\label{Lemma1}
For any content c and time slot $t$, if $k$ helpers minimizes the
total cost for $t$, then for any $t^\prime>t$, the total cost of using
$k^\prime>k$ helpers is higher than using $k$ helpers.
\end{lemma}

\begin{proof}
The total cost for content $c$ and time slot $t$ is:

\vspace{-2mm}
\begin{equation}
\label{delta}
\begin{aligned}
\Delta(x_{ct})=\sum_{r\in \mathcal{R}}w_{rc}e^{-x_{ct}\lambda\delta}+\alpha f(t) x_{ct}
\end{aligned}
\end{equation}
\vspace{-2mm}

As the minimum of $\Delta(x_{ct})$ occurs for $x_{ct}=k$, we have:
\vspace{-1mm}
\begin{equation}
\label{refeq11}
\vspace{-2mm}
\begin{aligned}
\sum_{r\in \mathcal{R}}w_{rc}e^{-k\lambda\delta}-\sum_{r\in \mathcal{R}}w_{rc} e^{-x^\prime_{ct}\lambda\delta}<\alpha f(t)(x^\prime_{ct}-k),
\end{aligned}
\end{equation}

for $x^\prime_{ct}>k$.
Also, $f(t)$ is an increasing function. Thus:

\begin{equation}
\label{refeq1}
\vspace{-2mm}
\begin{aligned}
\alpha f(t)(x^\prime_{ct^\prime}-k)<\alpha f(t^\prime)(x^\prime_{ct^\prime}-k),
\end{aligned}
\end{equation}
\vspace{-0mm}

for $x^\prime_{ct^\prime}>k$. From (\ref{refeq11}) and (\ref{refeq1}), we obtain:
\begin{equation}
\vspace{-2mm}
\label{bla}
\begin{aligned}
\sum_{r\in \mathcal{R}}w_{rc} e^{-k\lambda\delta}-\sum_{r\in \mathcal{R}}w_{rc} e^{-x^\prime_{ct^\prime}\lambda\delta}<\alpha f(t^\prime)(x^\prime_{ct^\prime}-k),
\end{aligned}
\end{equation}

for $x^\prime_{ct^\prime}>k.$
By rearranging the terms of \eqref{bla},
we have $\Delta(x_{ct^\prime})<\Delta(x^\prime_{ct^\prime})$
for $x_{ct^\prime}=k$ and $x^\prime_{ct^\prime}>k$.
\end{proof}

Next, we prove that for a content, if the initial number of helpers is
given, the optimum
can be obtained in polynomial
time.
The procedure is in
Algorithm~\ref{alg1}, see Lines~\ref{linegreedy1} and
\ref{linegreedy2}.

\begin{algorithm}
\caption{Optimization for given initial conditions}
\label{alg1}
\begin{algorithmic}[1]
\REQUIRE $C$, $H$, $T$
\ENSURE $\mathbf{z}$

\FOR{$k$ = $1$~to~$C$}
\STATE $\mathbf{z_k} \leftarrow [0]_{(H+1)}$ \label{initial}
 \FOR{$x_{k1}$ = $0$~to~$H$}
 \STATE $x^*_{k1}\leftarrow$ $x_{k1}$
   \STATE ${z}_k(x_{k1})\leftarrow$ $\Delta(x_{k1})$

   \FOR{$t$ = $2$~to~$T$} \label{linegreedy1}
\STATE $x^*_{kt} \leftarrow \underset{x_{kt}\in \{0,1,\dots,x^*_{k(t-1)}\}} {\argmin \{\Delta(x_{kt})\}}$ \label{linegreedy2}
\STATE $z_k(x_{k1}) \leftarrow z_k({x_{k1}})+\underset{x_{kt}\in \{0,1,\dots,x^*_{k(t-1)}\}} {\min \{\Delta(x_{kt})\}~~~~~}$ \label{linegreedy3}
\ENDFOR
\ENDFOR
\ENDFOR
\RETURN $\mathbf{z}$
\end{algorithmic}
\end{algorithm}
\vspace{-2mm}

In the algorithm, for each content $k\in \mathcal{C}$, a vector $\mathbf{z_k}$ of size
$(H+1)$ is used to store the optimal cost for all
possible initial numbers of helpers. Line $3$ considers all possible
initial numbers of helpers in the range $[0,H]$.  Lines
\ref{linegreedy1} and \ref{linegreedy2} compute the optimal values of
$x_{c2},\dots,x_{cT}$, denoted by $x^*_{c2},\dots,x^*_{cT}$
respectively, for given $x_{c1}$. The computation is of complexity
$O(HTR)$. For any content $c$ and any possible initial number of
helpers $h$, the optimal cost over all the slots are saved in
${z_c(h)}$. The overall complexity of Algorithm~\ref{alg1} is
$O(\max \{CH^2T,CTHR\})$, given $\Delta$ computed.

Note that Lines~\ref{linegreedy1} and~\ref{linegreedy2} are greedy by
construction. Namely, for time slot $t$, Line~\ref{linegreedy2}
determines the number of helpers that minimizes the cost of that
specific time slot. Even though this is intuitive, it is not obvious
that the greedy choice is globally optimal for the given initial
number of helpers. The optimality analysis is formalized in
Lemma~\ref{lemma2}.

\begin{lemma}\label{lemma2}
For any content $c \in \mathcal{C}$, if $x_{c1}$ is given, then
the optimal values of $x_{ct}$ i.e., $x^*_{ct}$, $t=\{2,3,\dots,T\}$
are computed via Lines \ref{linegreedy1} and \ref{linegreedy2} in
Algorithm~$1$ in polynomial time.
\end{lemma}
\begin{proof}
For any $c$, consider $x^*_{c2},\dots,x^*_{cT}$, the numbers of
helpers over the time slots returned by Algorithm~$1$ when $x_{c1}$ is
given. Consider another sequence $x^\prime_{c2},\dots,x^\prime_{cT}$
that differs from the first sequence and offers a lower total cost.
First consider the case that $x^*_{ct}<x^\prime_{ct}$ for some $t \in
\{2,\dots,T\}$, and $x^*_{ct}$ remains smaller than the values of the
second sequence in consecutive time slots until time slot $t+n$ where
$0 \le n \le T-t$. That is, the second sequence has elements
$x^\prime_{ct},~x^\prime_{c(t+1)}, \dots,~x^\prime_{c(t+n)}$ all being
greater than $x^*_{ct}$, whereas for slot $t+n+1$,
$x^\prime_{c(t+n+1)}\le x^*_{ct}$. Consider changing all of
$x^\prime_{ct}$, $x^\prime_{c(t+1)}$, \dots, $x^\prime_{c( t+n)}$ to
$x^*_{ct}$ in sequence two, while keeping the values of all other time
slots of this sequence. The updated sequence is feasible because
$x^\prime_{c(t+n+1)}\le x^*_{ct}$. Thus, monotonicity remains for the
updated sequence. The update reduces the cost of the second sequence
by Lemma~$1$, hence a contradiction. A
special case is $t+n=T$, for which $t+n+1$ does not exist. However the
same update and conclusion apply. One case remains, namely there is
no time slot $t$ with $x^*_{ct}<x^\prime_{ct}$, yet sequence two is
different from sequence one. In other words, $x^*_{c2}\ge
x^\prime_{c2}$, \dots, $x^*_{cT}\ge x^\prime_{cT}$.  Let $t$, $t \in
\{2,\dots,T\}$, be the first time slot with strict inequality, i.e.,
$x^*_{ct}>x^\prime_{ct}$. Such a time slot must exist, because,
otherwise the two sequences coincide. Consider increasing
$x^\prime_{ct}$ to $x^*_{ct}$. Sequence two remains feasible in terms of
being monotonically decreasing, because $x^\prime_{ct}\le
x^\prime_{c(t-1)}$ after setting $x^\prime_{ct}$ to $x^*_{ct}$ as
$x^*_{ct}\le x^*_{c(t-1)}=x^\prime_{c(t-1)}$. The cost of $t$, due to the
update, becomes lower because when $t$ is considered by the algorithm,
$x^*_{ct}$ is the optimum. Therefore in this case the second sequence
cannot be better either. Hence the result.
\end{proof}

By Algorithm~1, $x^*_{ct}$, $t=2,3,\dots,T$, can be computed if
$x_{c1}$, $c\in \mathcal{C}$, is given. Consequently, solving PRACO is
equivalent to finding the optimal values of $x_{c1}$, $c\in
\mathcal{C}$. We drop the second subscript and use $x_c, c\in
\mathcal{C}$ as optimization variables for the initial numbers of
helpers, and reformulate PRACO as follows. The cost of $x_{c}$, i.e.,
$z_c(x_{c})$ is from Algorithm~1. Constraint (\ref{const:capacity})
models the cache capacity.

\vspace{-2mm}
\begin{figure}[!ht]
\vskip -10pt
\begin{subequations}
\begin{alignat}{2}
\quad &
\min\limits_{x_{c}}\quad \sum_{c\in \mathcal{C}} z_c(x_{c}){x_{c}}
\label{F_e2} \\
\text{s.t}. \quad
&\sum_{c\in \mathcal{C}}{x_{c}}\le S \label{const:capacity}\\
& x_{c}\in \{0,1,\dots,H\}, c \in \mathcal{C}
\end{alignat}
\label{eq:reformulate}
\vskip -20pt
\end{subequations}
\end{figure}
\vspace{-5mm}

\section{The overall Algorithm and Optimality}\label{sec:DPalg}

\subsection{Dynamic Programming}

We use dynamic programming (DP) to obtain the
optimal values of $x_{c}$, $c\in \mathcal{C}$.
Denote by $a^*(k,i)$ the cost of optimal caching of the first
$k$ contents with a total cache capacity of $i$ units.
Thus, by definition, $a^*(C, S)$ is the overall optimal
cost. The values of $a^*(k,i), k=1,\dots,C, i=0,\dots, S$,
submit to recursion, as formalized in the lemma below.



\begin{lemma}\label{Lemma_Recursive}
The value of $a^*(k,i)$ can be derived from the recursive
formula shown in (\ref{RecursiveFunction1}) for $k=2, \dots, C$, with:
\end{lemma}
\vspace{-5mm}
\begin{equation}
a^*(1, i) =\underset{x_1\in\{0,1,\dots,\min\{i,H\}\}}  {\min\{z_1(x_{1})\},~~~~~~}
\end{equation}
\vspace{-3mm}
\begin{equation}
\label{RecursiveFunction1}
a^*(k,i)= \underset{x_k\in\{0,1,\dots,\min\{i,H\}\}~~~~~~~~~~~~~~~~~} {\min\{z_k(x_{k})+a^*(k-1,i-x_{k})\}}
\end{equation}

\begin{proof}
We use induction.  For $k=1$, the result is obvious for any $i=0, 1,
\dots, H$. Suppose that $a^*(k,i)$ is the optimal value for some
$k$, with $i$ in any range of interest.
By (\ref{RecursiveFunction1}), we have:

\vspace{-6mm}
\[
a^*(k+1,i)=\underset{x_{k+1}\in\{0,1,\dots,\min\{i,H\}\}~~~~~~~~~~~~~~~~~} {\min\{z_{k+1}(x_{k+1})+a^*(k,i-x_{k+1})\}}
\]
For $k+1$, the initial number of helpers $x_{k+1}$ must be
one of the values in $\{0,1,\dots,\min\{i,H\}$. For any value of
$x_{k+1}$, $z_{k+1}(x_{k+1})$ is the optimal cost for content $k+1$
(Lemma~\ref{lemma2}), and the corresponding cache capacity for contents
up to $k$ is $i-x_{k+1}$. For the latter,
$a^*(k,i-x_{k+1})$ is optimal. These, together with the $\min$-operator
give the optimum for $k+1$.
\end{proof}

\subsection{Algorithm Description and Optimality}
Algorithm~\ref{alg2} describes the DP approach. The input parameters
consist of $\mathbf{z}$, $C$, $S$, and $H$. Here, $\mathbf{z}$ is from
the output of Algorithm \ref{alg1}. Apart from $\bm{a}^*$ as defined
earlier, $\bm{b}^*$ is used to store the optimal caching solution.
Lines $3$-$10$ compute $a^*(k,i)$ and $b^*(k,i)$ for $k<C$, whereas
Lines $12$ and $13$ compute $a^*(C,S)$ and $b^*(C,S)$ for $k=C$.
Finally, $\bm{b}^*$ is mapped to optimal values of $\bm{x}$, denoted by $\bm{x^*}$, using Lines $14$-$20$.

\begin{algorithm}
\caption{The DP algorithm}
\label{alg2}
\begin{algorithmic}[1]
\REQUIRE $\mathbf{z}$, $C$, $S$, $H$
\ENSURE $\bm{x^*}$
\STATE $\bm{a^*} \leftarrow [0]_{C\times(S+1)}$,~$\bm{b^*} \leftarrow [0]_{C\times(S+1)}$,~$\bm{x^*}\leftarrow [0]_{C}$ \label{ini} 
\FOR{$k=1$ : $C$}
\IF {$k<C$}
\FOR{$i=0$ : $S$}
\IF {$k=1$}
\STATE $a^*(1,i)\leftarrow\underset{x_1\in\{0,1,\dots,\min\{i,H\}\}}{\min\{z_1({x_{1}})\}~~~~~~~}$
\STATE  $b^*(1,i) \leftarrow\underset{x_1\in\{0,1,\dots,\min\{i,H\}\}}{\argmin\{z_1({x_{1}})\}~~~}$
\ELSE
\STATE $a^*(k,i) \leftarrow \underset{x_k\in\{0,1,\dots,\min\{i,H\}\}~~~~~~~~~~~~~~~~} {\min\{z_k({x_{k}})+a^*(k-1,i-x_k)\}}$
\vspace{-2mm}
\STATE $b^*(k,i)~\leftarrow~\underset{x_k\in\{0,1,\dots,\min\{i,H\}\}~~~~~~~~~~~~~~~~~~~~~} {\argmin\{z_k({x_{k}})+a^*(k-1,i-x_k)\}}$
\ENDIF
\ENDFOR
\ELSE
\STATE $a^*(C,S)~\leftarrow~\underset{x_C\in\{0,1,\dots,H\}~~~~~~~~~~~~~~~~~~~~~~~~~~~}{\min\{z_C({x_{C}})+a^*(C-1,S-x_C)\}}$
\vspace{-2mm}
\STATE $b^*(C,S)~\leftarrow~\underset{~x_C\in\{0,1,\dots,H\}~~~~~~~~~~~~~~~~~~~~~~~~~~~~~~~~}{\argmin\{z_C({x_{C}})+a^*(C-1,S-x_C)\}}$
\ENDIF
\ENDFOR
\FOR{$k=C$ : $1$}
\IF {$k=C$}
\STATE ${x^*_{C}} \leftarrow {b^*(C,S)}$
\STATE $e \leftarrow S-b^*(C,S)$
\ELSE
\STATE ${x^*_{k}} \leftarrow {b^*(k,e)}$
\STATE $e \leftarrow e-b^*(k,e)$
\ENDIF
\ENDFOR
\RETURN $\bm{x^*}$
\end{algorithmic}
\end{algorithm}
\vspace{-2mm}
\begin{theorem}\label{theorem_dp}
Algorithm~\ref{alg2} delivers the global optimum of PRACO in polynomial time.
\end{theorem}
\begin{proof}
The optimality follows from Lemma~\ref{lemma2} and the recursion of
which the correctness is established in Lemma~\ref{Lemma_Recursive}.
As for time complexity, the steps in Algorithm~\ref{alg2} together
require a complexity of $O(HCS) = O(H^2Cs)$. However, a prerequisite
is that the $z$-values are given. Computing these values with
Algorithm~\ref{alg1}, given ${\Delta}$ computed, has complexity $O(\max \{CH^2T,CTHR\})$.
Hence the overall complexity is of $O(\max \{H^2Cs, CH^2T,CTHR\})$. Finally, note that, even though $s$ is
not a parameter for input size, its values is bounded by $C$, because
otherwise the capacity constraint is redundant and the problem
decomposes by content (and solved without the need of DP).
Hence the complexity is of $O(\max \{H^2C^2, CH^2T,CTHR\})$, which is
polynomial in input size.
\end{proof}
\vspace{-3mm}
\section{Performance Evaluation}

We compare the DP algorithm to two conventional caching algorithms, i.e., random caching \cite{Optimalgeographic2015} and popular caching \cite{Videoaware2014}. Both algorithms consider contents for caching one by one. In the former, the contents are considered randomly, but with respect to the files' request probabilities; a content with higher request probabilities will be more likely selected for caching. In the latter, popular contents, i.e., contents with higher request probabilities, will be considered first. For the content under consideration, the cache decision is the number of helpers with minimum total cost.

We use a Zipf distribution with shape parameter $\gamma$ to
characterize the content request probability for any requester.
Thus, $w_{rc}=\frac{c^{-\gamma}}{\sum_{k\in
\mathcal{C}}k^{-\gamma}}, r \in {\mathcal R}$.
Same as~\cite{Shukla2017Proactive}, the time period is set to $24$ hours,
and the duration of each time slot ($\delta$) is $1$ hour. The storage cost
is simulated using $f(t)=t^2$.

Figures~\ref{fig:impactH}-\ref{fig:impactgamma} provide the results
and show the impacts of parameters $H$, $\alpha$, $s$, and $\gamma$ on
the cost, respectively. It can be seen that the cost decreases with
respect to all the mentioned parameters. This is quite expected. For
example, when $H$ increases, the requesters have more opportunity to
meet helpers, leading to lower cost.  The same conclusion can be made for cheaper storage (small $\alpha$), and higher capacity (larger $s$). For parameter $\gamma$, a higher value means more variation in the
contents' request probabilities, thus it is easier for the algorithms to
identify caching solutions such that the helpers more likely store
the requested contents.

The DP algorithm outperforms the two conventional caching algorithms.
In Figures 2-4, the improvement is significant when $H$ and $\delta$
increase and $\alpha$ decreases. For example, when $H$ increases from
$4$ to $20$, the DP algorithm outperforms the popular caching
algorithm by $13\%$ to $24\%$, and outperforms the random caching
algorithm by $27\%$ to $35\%$.  This is because the DP algorithm uses
the storage capacity of helpers optimally in comparison to the
conventional algorithms.

Recall that small $\alpha$ means low storage cost. When
$\alpha=0.01$, which is a fairly large value in the context, the optimal
strategy tends to not to store contents --
it is more preferable to download from the server.
Hence cache optimization is less relevant and the algorithms are
similar in performance.  When $\alpha$ decreases, the
difference between the performance of the DP and the other algorithms becomes apparent, as the DP algorithm uses the storage capacity optimally while the conventional algorithms are not able to accomplish this.

\vspace{-3mm}
\makeatletter
\setlength{\@fptop}{0pt}
\makeatother
\begin{figure}[ht!]
\begin{minipage}{.4\textwidth}
\centering
\includegraphics[scale=0.5]{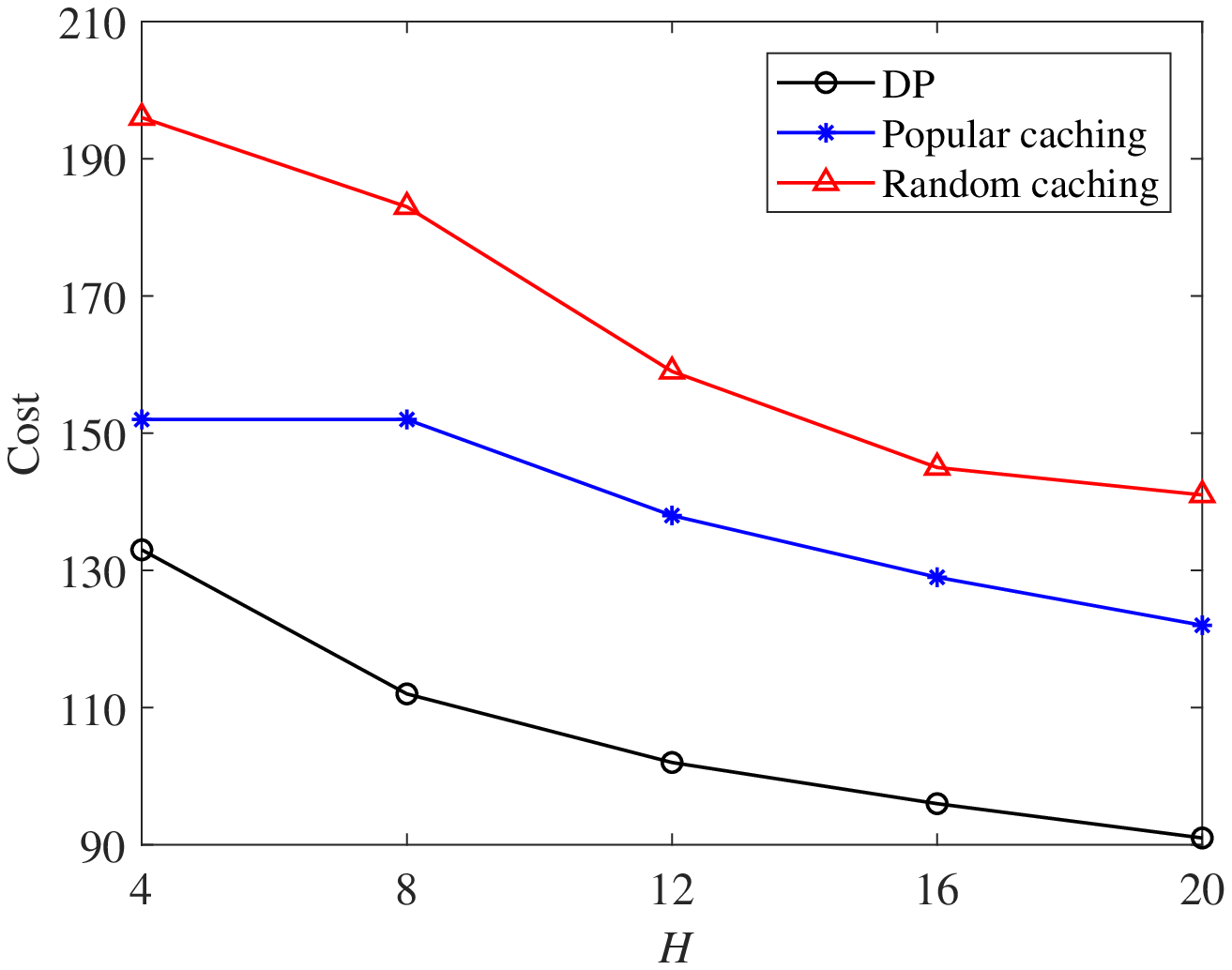}
\vspace{-8mm}
\caption{Impact of $H$ on cost when $s=4$, $C=100$, $T=24$, $\delta=1$, $R=10$, $\gamma=1$, $\lambda=1$,\text{ and} $\alpha=0.0001$.}
 \label{fig:impactH}
\end{minipage}

\begin{minipage}{.4\textwidth}
\centering
\includegraphics[scale=0.5]{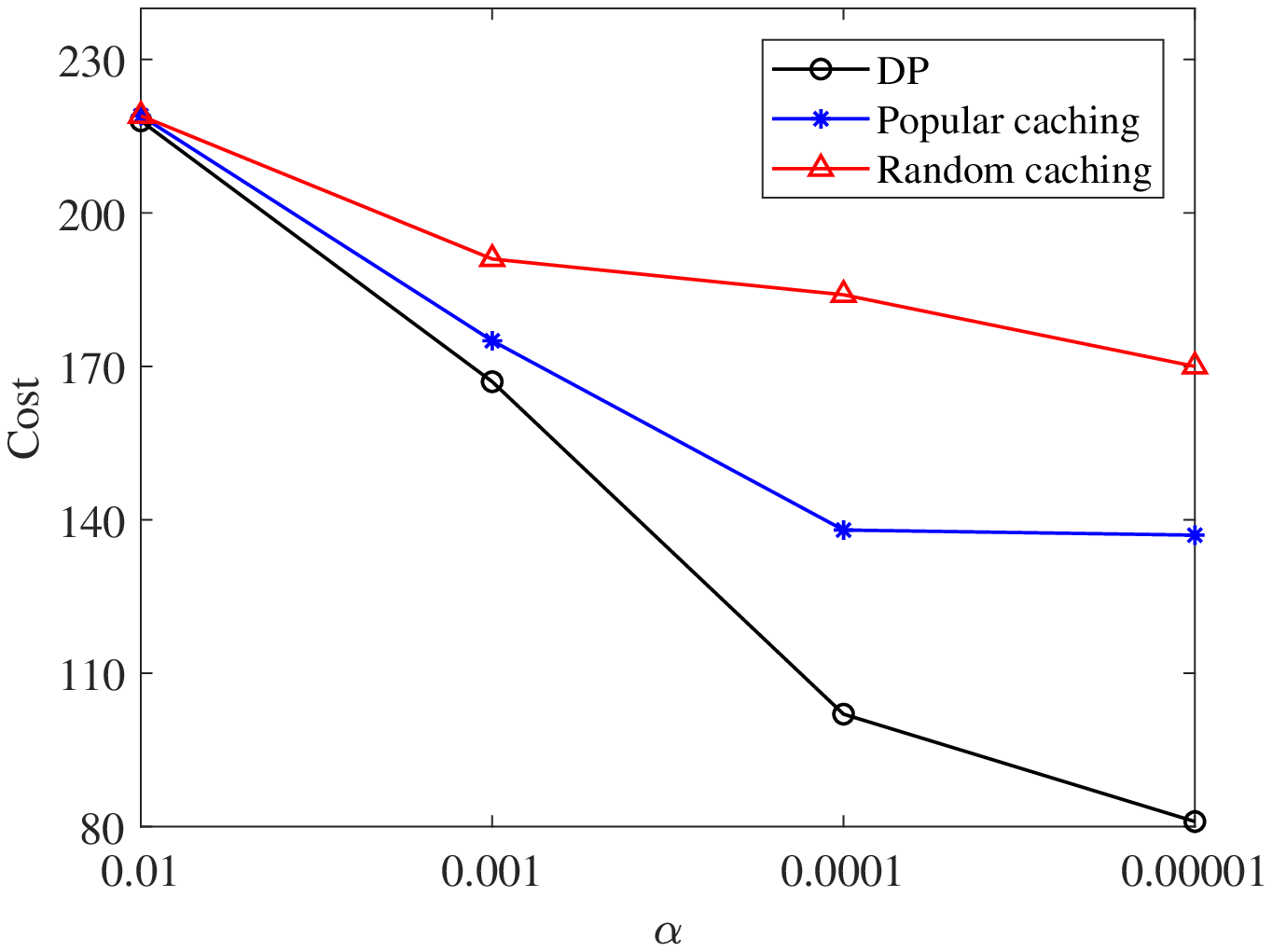}
\vspace{-8mm}
\begin{center}
  \caption{Impact of $\alpha$ on cost when $H=12$, $s=4$, $C=100$, $T=24$, $\delta=1$, $R=10$, $\lambda=1$,\text{ and} $\gamma=1$.}
  \label{fig:impactalpha}
\end{center}
\end{minipage}
\end{figure}
\vspace{-2mm}

\makeatletter
\setlength{\@fptop}{0pt}
\makeatother
\begin{figure}[ht!]
\begin{minipage}{.4\textwidth}
\centering
\includegraphics[scale=0.5]{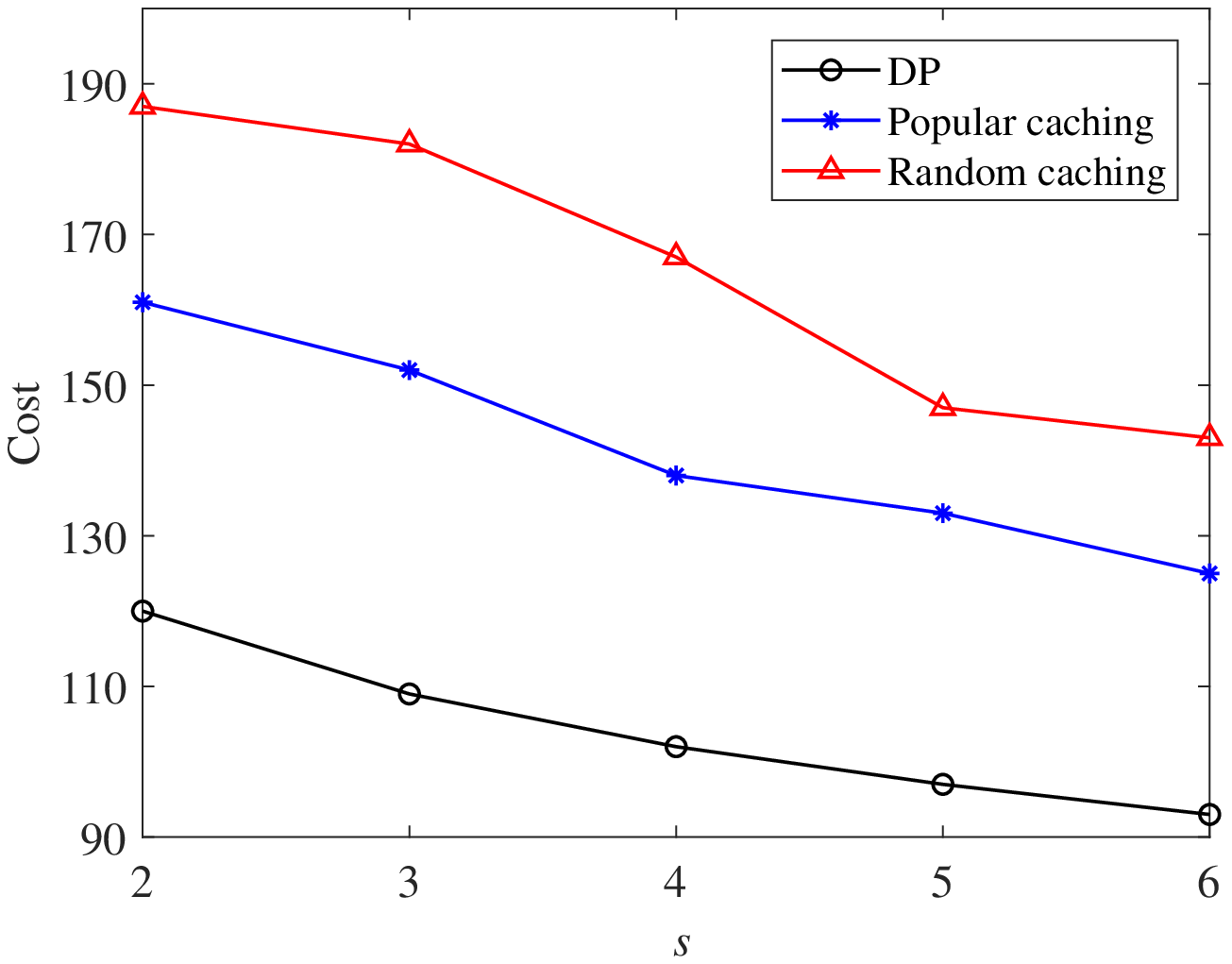}
\vspace{-8mm}
\begin{center}
  \caption{Impact of $s$ on cost when $H=12$, $C=100$, $T=24$, $\delta=1$, $R=10$, $\gamma=1$, $\lambda=1$,\text{ and} $\alpha=0.0001$.}
  \label{fig:impactcachesize}
\end{center}
\end{minipage}

\begin{minipage}{.4\textwidth}
\centering
\includegraphics[scale=0.5]{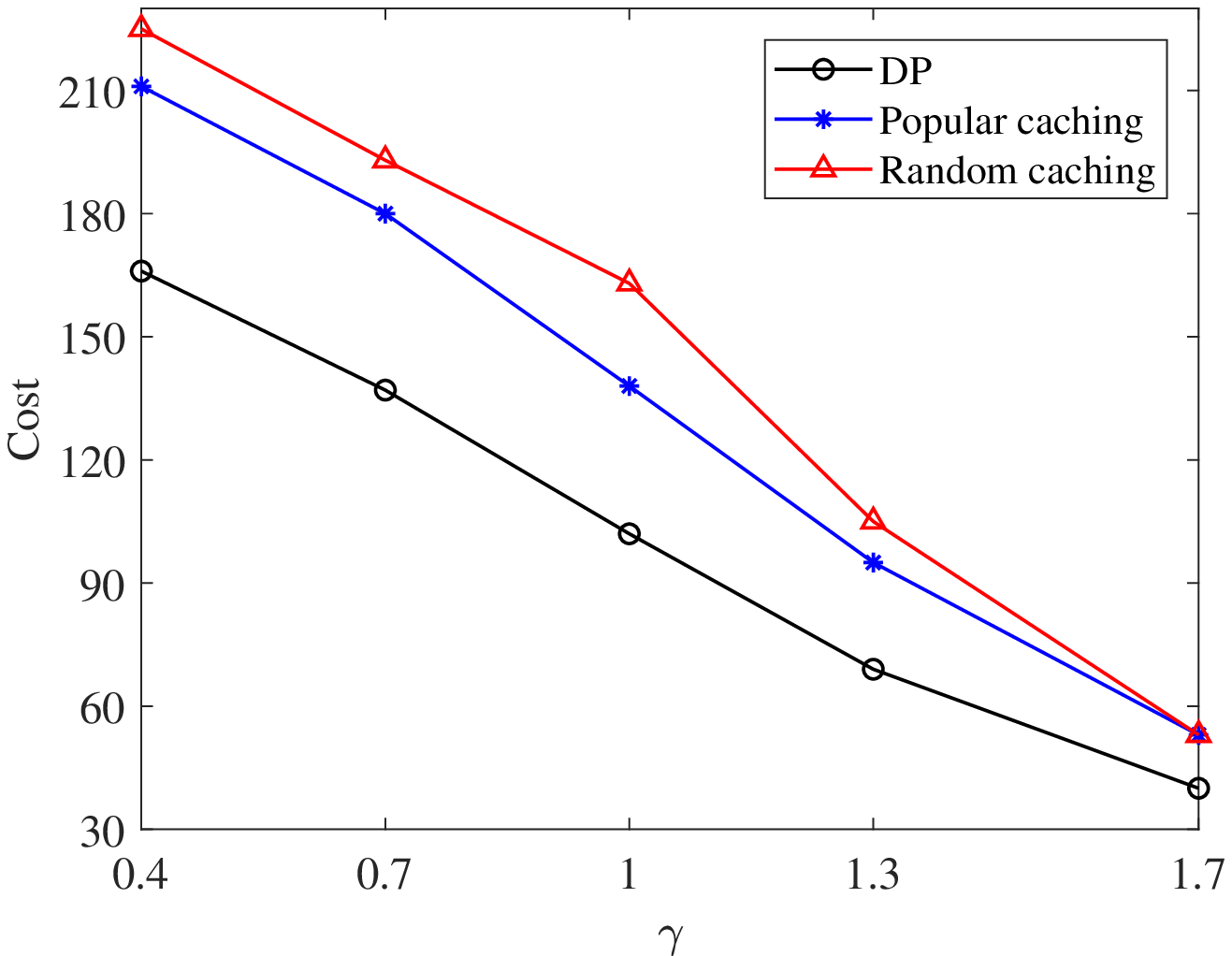}
\vspace{-8mm}
\begin{center}
  \caption{Impact of $\gamma$ on cost when $H=12$, $s=4$, $C=100$, $T=24$, $\delta=1$, $R=10$, $\lambda=1$,\text{ and} $\alpha=0.0001$.}
  \label{fig:impactgamma}
\end{center}
\end{minipage}
\end{figure}
\vspace{-5mm}

\section{Conclusions}

The paper has studied a proactive retention-aware caching problem, considering user mobility, storage cost, and cache size. We have provided analysis and algorithm development, proving that global optimum is within reach in polynomial time. Simulation results have manifested significant improvements by the proposed algorithm in comparison to two conventional caching algorithms. In our future work, we consider a more general system scenario including non-homogeneous contact rates, helpers with different cache sizes, and contents with different sizes. Thus, the problem becomes more challenging and new solution approaches need to be developed.
\bibliographystyle{IEEEtran}
\bibliography{IEEEabrv,ForIEEEBib}

 \end{document}